\documentclass[10pt,conference,letterpaper]{IEEEtran}
\usepackage{times,amsmath,epsfig,amssymb}

\usepackage{graphicx}
\usepackage{balance}  %
\usepackage{type1cm}     %
\usepackage{algorithm}     %
\usepackage{algorithmic}     %
\usepackage{xspace}     %
\usepackage{booktabs}     %
\usepackage[bf,tableposition=top]{caption}     %
\usepackage{siunitx}          %
\usepackage[hyphens]{url}     %
\DeclareGraphicsExtensions{.pdf,.png,.jpg,.eps}
\graphicspath{{./img/}}

\usepackage[square,numbers]{natbib}     %
\newcommand{\spara}[1]{\smallskip\noindent\textbf{#1}}

\newenvironment {squishlist}
{\begin{list}{$\bullet$}
  { \setlength{\itemsep}{1pt}
     \setlength{\parsep}{1pt}
     \setlength{\topsep}{1pt}
     \setlength{\partopsep}{1pt}
     \setlength{\leftmargin}{1.5em}
     \setlength{\labelwidth}{1em}
     \setlength{\labelsep}{0.5em} } }
{\end{list}}

\newcommand{\tuple}[1]{\ensuremath{\langle #1 \rangle}\xspace}
\newcommand{\hash}[1]{\ensuremath{\mathcal{H}_{#1}}\xspace}
\newcommand{\pei}{\textsc{pei}\xspace}
\newcommand{\peis}{{\pei}s\xspace}
\newcommand{\pe}{\textsc{pe}\xspace}
\newcommand{\pes}{{\pe}s\xspace}
\newcommand{\potc}{\textsc{p\textup{o}tc}\xspace}
\newcommand{\dspe}{\textsc{dspe}\xspace}
\newcommand{\dspes}{{\dspe}s\xspace}
\newcommand{\dagr}{\textsc{dag}\xspace}

\newcommand{\pkg}{\textsc{Partial Key Grouping}\xspace}
\newcommand{\pkgs}{\textsc{pkg}\xspace}
\newcommand{\kg}{\textsc{kg}\xspace}
\newcommand{\sg}{\textsc{sg}\xspace}
\newcommand{\sources}{\ensuremath{\mathcal{S}}\xspace}
\newcommand{\numsources}{S\xspace}
\newcommand{\workers}{\ensuremath{\mathcal{W}}\xspace}
\newcommand{\numworkers}{W\xspace}
\newcommand{\keyspace}{\ensuremath{\mathcal{K}}\xspace}
\newcommand{\keysize}{K\xspace}
\newcommand{\mycomment}[1]{}

\DeclareMathOperator*{\expect}{\mathbb{E}}
\DeclareMathOperator*{\avg}{avg}
\DeclareMathOperator*{\argmin}{argmin}

\newtheorem{theorem}{Theorem}[section]

\newtheorem{lemma}[theorem]{Lemma}

\newtheorem{corollary}[theorem]{Corollary}

\title{The Power of Both Choices: Practical Load Balancing~for~Distributed~Stream~Processing~Engines}
\author{%
{Muhammad Anis Uddin Nasir{\small $^{\#1}$},
Gianmarco De Francisci Morales{\small $^{*2}$},
David Garc\'ia-Soriano{\small $^{*3}$}}\\
{Nicolas Kourtellis{\small $^{*4}$},
Marco Serafini{\small $^{\$5}$} }
\vspace{1.6mm}\\
\fontsize{10}{10}\selectfont\itshape
$^{\#}$KTH Royal Institute of Technology, Stockholm, Sweden\\
$^{*}$Yahoo Labs, Barcelona, Spain\\
$^{\$}$Qatar Computing Research Institute, Doha, Qatar\\
\fontsize{9}{9}\selectfont\ttfamily\upshape
$^{1}$anisu@kth.se,
$^{2}$gdfm@apache.org,
$^{3}$davidgs@yahoo-inc.com\\
$^{4}$kourtell@yahoo-inc.com,
$^{5}$mserafini@qf.org.qa
}
\begin{document}
\maketitle
\begin{abstract}

We study the problem of load balancing in distributed stream processing engines, which is exacerbated in the presence of skew.
We introduce \pkg (\pkgs), a new stream partitioning scheme that adapts the classical ``power of two choices'' to a distributed streaming setting by leveraging two novel techniques: \emph{key splitting} and \emph{local load estimation}.
In so doing, it achieves better load balancing than key grouping while being more scalable than shuffle grouping.

We test \pkgs on several large datasets, both real-world and synthetic.
Compared to standard hashing, \pkgs reduces the load imbalance by up to several orders of magnitude, and often achieves nearly-perfect load balance.
This result translates into an improvement of up to 60\% in throughput and up to 45\% in latency when deployed on a real Storm cluster.
\end{abstract}

\section{Introduction}
\label{sec:intro}

Distributed stream processing engines (\dspes) such as S4,\footnote{\url{https://incubator.apache.org/s4}} Storm,\footnote{\url{https://storm.incubator.apache.org}} and Samza\footnote{\url{https://samza.incubator.apache.org}} have recently gained much attention owing to their ability to process huge volumes of data with very low latency on clusters of commodity hardware.
Streaming applications are represented by directed acyclic graphs (\dagr) where vertices, called \emph{processing elements} (\pes), represent operators, and edges, called \emph{streams}, represent the data flow from one \pe to the next.
For scalability, streams are partitioned into sub-streams and processed in parallel on a replica of the \pe called \emph{processing element instance} (\pei).

Applications of \dspes, especially in data mining and machine learning, typically require accumulating state across the stream by grouping the data on common fields \cite{ben-haim2010spdt,berinde2010heavyhitters}.
Akin to MapReduce, this grouping in \dspes is usually implemented by partitioning the stream on a \emph{key} and ensuring that messages with the same key are processed by the same \pei.
This partitioning scheme is called \emph{key grouping}.
Typically, it maps keys to sub-streams by using a hash function.
Hash-based routing allows each source \pei to route each message solely via its key,  without needing to keep any state or to coordinate among \peis.
Alas, it also results in load imbalance as it represents a ``single-choice'' paradigm~\citep{one_choice_load}, and because it disregards the popularity of a key, i.e., the number of messages with the same key in the stream, as depicted in Figure~\ref{fig:imbalance}.

Large web companies run massive deployments of \dspes in production.
Given their scale, good utilization of the resources is critical.
However, the skewed distribution of many workloads causes a few \peis to sustain a significantly higher load than others. %
This suboptimal load balancing leads to poor resource utilization and inefficiency.

Another partitioning scheme called \emph{shuffle grouping} achieves excellent load balancing by using a round-robin routing, i.e., by sending a message to a new \pei in cyclic order, irrespective of its key.
However, this scheme is mostly suited for stateless computations.
Shuffle grouping may require an additional aggregation phase and more memory to express stateful computations (Section~\ref{sec:preliminaries}).
Additionally, it may cause a decrease in accuracy for data mining algorithms (Section~\ref{sec:applications}).

\begin{figure}[t]
\begin{center}
\includegraphics[scale=0.3]{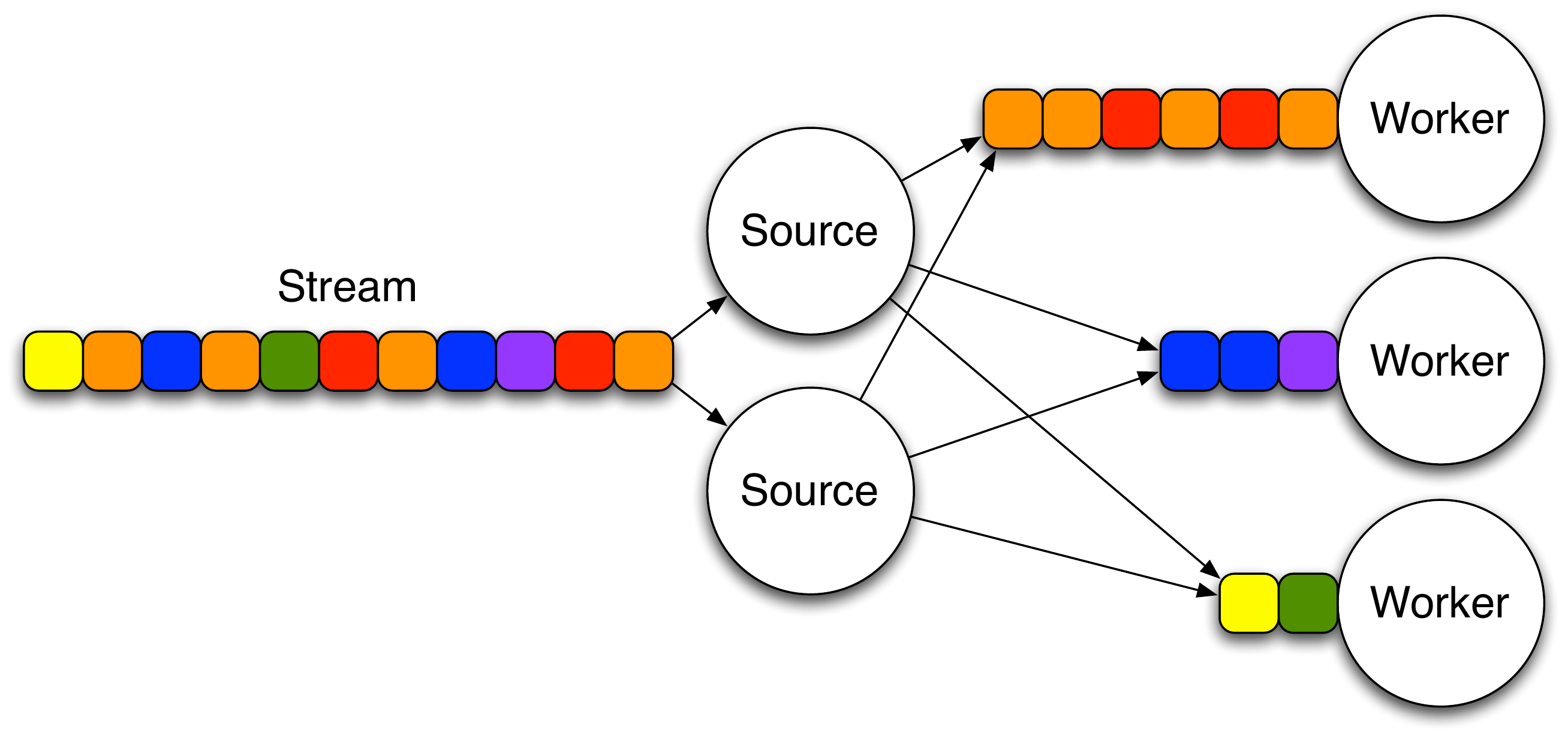}
\caption{Load imbalance generated by skew in the key distribution when using key grouping. %
The color of each message represents its key.}
\label{fig:imbalance}
\end{center}
\end{figure}

In this work, we focus on the problem of load balancing of stateful applications in \dspes when the input stream follows a skewed key distribution.
In this setting, load balancing is attained by having upstream \peis create a balanced partition of messages for  downstream \peis, for each edge of the \dagr.
Any practical solution for this task needs to be both \emph{streaming} and \emph{distributed}: the former constraint enforces the use of an online algorithm, as the distribution of keys is not known in advance, while the latter calls for a decentralized solution with minimal coordination overhead in order to ensure scalability.

\clearpage

To address this problem, we leverage the ``power of two choices''~\citep{mitzenmacher2001power} (\potc), whereby the system picks the least loaded out of two candidate \peis for each key.
However, to maintain the semantics of key grouping while using \potc (i.e., so that one key is handled by a single \pei), sources would need to track which of the two possible choices has been made for each key.
This requirement imposes a coordination overhead every time a new key appears, so that all sources agree on the choice.
In addition, sources should then store this choice in a routing table. 
Each edge in the \dagr would thus require a routing table for every source, each with one entry per key.
Given that a typical stream may contain billions of keys, this solution is not practical.

Instead, we propose to relax the key grouping constraint and allow each key to be handled by \emph{both} candidate \peis.
We call this technique \emph{key splitting}; it allows us to apply \potc without the need to agree on, or keep track of, the choices made.
As shown in Section~\ref{sec:evaluation}, key splitting guarantees good load balance even in the presence of skew.

A second issue is how to estimate the load of a downstream \pei.
Traditional work on \potc assumes global knowledge of the current load of each server, which is challenging in a distributed system.
Additionally, it assumes that all messages originate from a single source, whereas messages in a \dspe are generated in parallel by multiple sources.

In this paper we prove that, interestingly, a simple \emph{local load estimation} technique, whereby each source independently tracks the load of downstream \peis, performs very well in practice.
This technique gives results that are almost indistinguishable from those given by a global load oracle.
The combination of these two techniques (key splitting and local load estimation) enables a new stream partitioning scheme named \emph{\pkg}.

In summary, we make the following contributions.
\begin{squishlist}
\item We study the problem of load balancing in modern distributed stream processing engines.
\item We show how to apply \potc to \dspes in a principled and practical way, and propose two novel techniques to do so: key splitting and local load estimation.
\item We propose \pkg, a novel and simple stream partitioning scheme that applies to any \dspe.
When implemented on top of Apache Storm, it requires a single function and less than 20 lines of code.\footnote{Available at \url{https://github.com/gdfm/partial-key-grouping}}
\item We measure the impact of \pkgs on a real deployment on Apache Storm.
Compared to key grouping, it improves the throughput of an example application on real-world datasets by up to $60\%$, and the latency by up to $45\%$.
\end{squishlist}

\section{Preliminaries and Motivation}
\label{sec:preliminaries}

We consider a \dspe running on a cluster of machines that communicate by exchanging messages following the flow of a \dagr, as discussed.
In this work, we focus on balancing the data transmission along a single edge in a \dagr.
Load balancing across the whole \dagr is achieved by balancing along each edge independently.
Each edge represents a single stream of data, along with its partitioning scheme.
Given a stream under consideration, let the set of upstream \peis (sources) be \sources, and the set of downstream \peis (workers) be \workers, and their sizes be $|\sources| = \numsources$ and $|\workers| = \numworkers$ (see Figure~\ref{fig:imbalance}).
The input to the engine is a sequence of messages $m = \tuple{t,k,v}$ where $t$ is the timestamp at which the message is received, $k \in \keyspace \,, |\keyspace| = \keysize$ is the message key, and $v$ is the value.
The messages are presented to the engine in ascending order by timestamp.

A {\em stream partitioning} function $P_{t}: \mathcal{K} \rightarrow \mathbb{N}$ maps each key in the key space to a natural number, at a given time $t$.
This number identifies the worker responsible for processing the message.
Each worker is associated to one or more keys.
We use a definition of \emph{load} similar to others in the literature (e.g., Flux~\citep{shah2003flux}).
At time $t$, the load of a worker $i$ is the number of messages handled by the worker up to $t$:
$$ L_i(t) = |\{ \tuple{\tau,k,v} : P_{\tau}(k) = i \wedge \tau \leq t \}| $$

In principle, depending on the application, two different messages might impose a different load on workers.
However, in most cases these differences even out and modeling such application-specific differences is not necessary.

We define {\em imbalance} at time $t$ as the difference between the maximum and the average load of the workers:
$$ I(t) = \max_{i}(L_i(t)) - \avg_{i}(L_i(t)), \text{ for } i \in \workers$$

We tackle the problem of identifying a stream partitioning function that minimizes the imbalance,
while at the same time avoiding the downsides of shuffle grouping.

\subsection{Existing Stream Partitioning Functions}
\label{sec:existing-partitioning}
Data is sent between \pes by exchanging messages over the network.
Several primitives are offered by \dspes for sources to partition the stream, i.e., to route messages to different workers.
There are two main primitives of interest: \emph{key grouping} (\kg) and \emph{shuffle grouping} (\sg).
\kg ensures that messages with the same key are handled by the same \pei (analogous to MapReduce).
It is usually implemented through hashing. %

\sg routes messages independently, typically in a round-robin fashion.
\sg provides excellent load balance by assigning an almost equal number of messages to each \pei.
However, no guarantee is made on the partitioning of the key space, as each occurrence of a key can be assigned to any \peis.
\sg is the perfect choice for \emph{stateless} operators.
However, with \emph{stateful} operators one has to handle, store and aggregate multiple partial results for the same key, thus incurring additional costs. %
In general, when the distribution of input keys is skewed, the number of messages that each \pei needs to handle can vary greatly.
While this problem is not present for stateless operators, which can use \sg to evenly distribute messages, stateful operators implemented via \kg suffer from load imbalance.
This issue generates a degradation of the service level, or reduces the utilization of the cluster which must be provisioned to handle the peak load of the single most loaded server.

\spara{Example.}
To make the discussion more concrete, we introduce a simple application that will be our running example: \emph{streaming top-k word count}.
This application is an adaptation of the classical MapReduce word count to the streaming paradigm where we want to generate a list of top-k words by frequency at periodic intervals (e.g., each $T$ seconds).
It is also a common application in many domains, for example to identify trending topics in a stream of tweets.

\spara{Implementation via key grouping.} Following the MapReduce paradigm, the implementation of word count described by~\citet{neumeyer2010s4} or~\citet{noll2013rolling} uses \kg on the source stream.
The counter \pe keeps a running counter for each word.
\kg ensures that each word is handled by a single \pei, which thus has the total count for the word in the stream.
At periodic intervals, the counter \peis send their top-k counters to a single downstream aggregator to compute the top-k words.
While this application is clearly simplistic, it models quite well a general class of applications common in data mining and machine learning whose goal is to create a model by tracking aggregated statistics of the data.
Clearly \kg generates load imbalance as, for instance, the \pei associated to the key ``the" will receive many more messages than the one associated with ``Barcelona".
This example captures the core of the problem we tackle: the distribution of word frequencies follows a Zipf law where few words are extremely common while a large majority are rare.
Therefore, an even distribution of keys such as the one generated by \kg results in an uneven distribution of messages.

\spara{Implementation via shuffle grouping.} An alternative implementation uses shuffle grouping on the source stream to get partial word counts.
These counts are sent downstream to an aggregator every $T$ seconds via key grouping.
The aggregator simply combines the counts for each key to get the total count and selects the top-k for the final result.
Using \sg requires a slightly more complex logic but it generates an even distribution of messages among the counter \peis.
However, it suffers from other problems.
Given that there is no guarantee which \pei will handle a key, each \pei potentially needs to keep a counter for \emph{every} key in the stream.
Therefore, the memory usage of the application grows linearly with the parallelism level.
Hence, it is not possible to scale to a larger workload by adding more machines: the application is not scalable in terms of memory.
Even if we resort to approximation algorithms, in general, the error depends on the number of aggregations performed, thus it grows linearly with the parallelism level.
We analyze this case in further detail along with other application scenarios in Section~\ref{sec:applications}.

\subsection{Key grouping with rebalancing}
One common solution for load balancing in \dspes\ is operator migration~\citep{shah2003flux,cherniack2003scalable,xing2005dynamic,gedik2013partitioning,balkesen2013adaptive,castro2013integrating}.
Once a situation of load imbalance is detected, the system activates a rebalancing routine that moves part of the keys, and the state associated with them, away from an overloaded server.
While this solution is easy to understand, its application in our context is not straightforward for several reasons.

Rebalancing requires setting a number of parameters such as how often to check for imbalance and how often to rebalance.
These parameters are often application-specific as they involve a trade-off between imbalance and rebalancing cost that depends on the size of the state to migrate. %

Further, implementing a rebalancing mechanism usually requires major modifications of the \dspe at hand.
This task may be hard, and is usually seen with suspicion by the community driving open source projects, as witnessed by the many variants of Hadoop that were never merged back into the main line of development~\citep{abouzeid2009hadoopdb,Dittrich2010hadooppp,yang2007mapreducemerge}.

In our context, rebalancing implies migrating keys from one sub-stream to another.
However, this migration is not directly supported by the programming abstractions of some \dspes.
Storm and Samza use a \emph{coarse-grained} stream partitioning paradigm.
Each stream is partitioned into as many sub-streams as the number of downstream \peis.
Key migration is not compatible with this partitioning paradigm, as a key cannot be uncoupled from its sub-stream. 
In contrast, S4 employs a \emph{fine-grained} paradigm where the stream is partitioned into one sub-stream per key value, and there is a one-to-one mapping of a key to a \pei.
The latter paradigm easily supports migration, as each key is processed independently.

A major problem with mapping keys to \peis explicitly is that the \dspe must maintain several routing tables: one for each stream. %
Each routing table has one entry for each key in the stream.
Keeping these tables is impractical because the memory requirements are staggering.
In a typical web mining application, each routing table can easily have billions of keys.
For a moderately large \dagr with tens of edges, each with tens of sources, the memory overhead easily becomes prohibitive.

Finally, as already mentioned, for each stream there are several sources sending messages in parallel.
Modifications to the routing table must be consistent across all sources, so they require coordination, which creates further overhead.
For these reasons we consider an alternative approach to load balancing. %

\section{Partial Key Grouping}
\label{sec:algos}

The problem described so far currently lacks a satisfying solution.
To solve this issue, we resort to a widely-used technique in the literature of load balancing: the so-called \emph{``power of two choices''} (\potc).
While this technique is well-known and has been analyzed thoroughly both from a theoretical and practical perspective~\citep{adler1995parallelrandomized,azar1999balanced-allocations,byers2003geometricgeneralizations,lenzen2011parallelrandomized,mitzenmacher2001power,mitzenmacher2001potc-survey}, its application in the context of \dspes is not straightforward and has not been previously studied.
Introduced by \citet{azar1999balanced-allocations}, \potc is a simple and elegant technique that allows to achieve load balance when assigning units of load to workers.
It is best described in terms of ``balls and bins''.
Imagine a process where a stream of balls (units of work) is distributed to a set of bins (the workers) as evenly as possible.
The \emph{single-choice paradigm} corresponds to putting each ball into one bin selected uniformly at random.
By contrast, the power of two choices selects two bins uniformly at random, and puts the ball into the least loaded one.
This simple modification of the algorithm has powerful implications that are well known in the literature (see Sections~\ref{sec:theory},~\ref{sec:rel-work}). %

\spara{Single choice.}
The current solution used by all \dspes to partition a stream with key grouping corresponds to the single-choice paradigm.
The system has access to a single hash function $\hash{1}(k)$.
The partitioning of keys into sub-streams is determined by the function $ P_{t}(k) = \hash{1}(k) \bmod W $,
where $\bmod$ is the modulo operator.

The single-choice paradigm is attractive because of its simplicity: the routing does not require to maintain any state and can be done independently in parallel.
However, it suffers from a problem of load imbalance~\citep{mitzenmacher2001power}.
This problem is exacerbated when the distribution of input keys is skewed.
\spara{PoTC.}
When using the power of two choices, we have two hash functions $\hash{1}(k)$ and $\hash{2}(k)$.
The algorithm maps each key to the sub-stream assigned to the least loaded worker between the two possible choices, that is:
$P_{t}(k)$~=~$\argmin_i(L_i(t) : \hash{1}(k) = i \lor \hash{2}(k) = i)$.

The theoretical gain in load balance with two choices is exponential compared to a single choice.
However, using more than two choices only brings constant factor improvements~\citep{azar1999balanced-allocations}.
Therefore, we restrict our study to two choices.

\potc introduces two additional complications.
First, to maintain the semantics of key grouping, the system needs to \emph{keep state} and track the choices made.
Second, the system has to \emph{know the load} of the workers in order to make the right choice.
We discuss these two issues next.

\subsection{Key Splitting}

A na\"{i}ve application of \potc to key grouping requires the system to store a bit of information for each key seen, to keep track of which of the two choices needs to be used thereafter.
This variant is referred to as \emph{static} \potc.

Static \potc incurs some of the problems discussed for key grouping with rebalancing.
Since the actual worker to which a key is routed is determined dynamically, sources need to keep a routing table with an entry per key.
As already discussed, maintaining this routing table is often impractical.

In order to leverage \potc and make it viable for \dspes, we relax the requirement of key grouping.
Rather than mapping each key to one of the two possible choices, we allow it to be mapped to \emph{both choices}.
Every time a source sends a message, it selects the worker with the lowest current load among the two candidates associated to that key.
This technique, called \emph{key splitting}, introduces several new trade-offs.

First, key splitting allows the system to operate in a decentralized manner, by allowing multiple sources to take decisions independently in parallel.
As in key grouping and shuffle grouping, no state needs to be kept by the system and each message can be routed independently.

Key splitting enables far better load balancing compared to key grouping.
It allows using \potc to balance the load on the workers:
by splitting each key on multiple workers, it handles the skew in the key popularity.
Moreover, given that \emph{all} its decisions are dynamic and based on the current load of the system (as opposed to static \potc), key splitting adapts to changes in the popularity of keys over time.

Third, key splitting reduces the memory usage and aggregation overhead compared to shuffle grouping.
Given that each key is assigned to \emph{exactly two} \peis, the memory to store its state is just a constant factor higher than when using key grouping.
Instead, with shuffle grouping the memory grows linearly with the number of workers $W$.
Additionally, state aggregation needs to happen only once for the two partial states, as opposed to $W - 1$ times in shuffle grouping.
This improvement also allows to reduce the error incurred during aggregation for some algorithms, as discussed in Section~\ref{sec:applications}.

From the point of view of the application developer, key splitting gives rise to  a novel stream partitioning scheme called \pkg, which lies in-between key grouping and shuffle grouping.

Naturally, not all algorithms can be expressed via \pkgs.
The functions that can leverage \pkgs are the same ones that can leverage a combiner in MapReduce, i.e., associative functions and monoids.
Examples of applications include na\"{i}ve Bayes, heavy hitters, and streaming parallel decision trees, as detailed in Section~\ref{sec:applications}.
On the contrary, other functions such as computing the median cannot be easily expressed via \pkgs.

\spara{Example.}
Let us examine the streaming top-k word count example using \pkgs.
In this case, each word is tracked by two counters on two different \peis. %
Each counter holds a partial count for the word, while the total count is the sum of the two partial counts.
Therefore, the total memory usage is $2 \times K$, i.e., $O(K)$.
Compare this result to \sg where the memory is $O(WK)$.
Partial counts are sent downstream to an aggregator that computes the final result.
For each word, the application sends two counters, and the aggregator performs a constant time aggregation.
The total work for the aggregation is $O(K)$.
Conversely, with \sg the total work is again $O(WK)$.
Compared to \kg, the implementation with \pkgs requires additional logic, some more memory and has some aggregation overhead.
However, it also provides a much better load balance which maximizes the resource utilization of the cluster.
The experiments in Section~\ref{sec:evaluation} prove that the benefits outweigh its cost.

\subsection{Local Load Estimation}
\potc requires knowledge of the load of each worker to take its routing decision.
A \dspe is a distributed system, and, in general, sources and workers are deployed on different machines.
Therefore, the load of each worker is not readily available to each source.

Interestingly, we prove that no communication between sources and workers is needed to effectively apply \potc.
We propose a \emph{local load estimation} technique, whereby each source independently maintains a local load-estimate vector with one element per worker.
The load estimates are updated by using only local information of the portion of stream sent by each source.
We argue that in order to achieve global load balance it is sufficient that each source independently balances the load it generates across all workers.

The correctness of local load estimation directly follows from our standard definition of load in Section~\ref{sec:preliminaries}.
The load on a worker $L_i$ is simply the sum of the loads that each source $j$ imposes on the given worker:
$ L_i(t) = \sum_{j \in \sources}{L_i^j(t)}. $
Each source $j$ can keep an estimate of the load on each worker $i$ based on the load it has generated $L_i^j$. %
As long as each source keeps its own portion of load balanced, then the overall load on the workers will also be balanced.
Indeed, the maximum overall load is at most the sum of the maximum load that each source sees locally. 
It follows that the maximum imbalance is also at most the sum of the local imbalances. %
\section{Analysis}
\label{sec:theory}

\newcommand{\naturals}{{\mathbb N}}
\DeclareRobustCommand{\calA}[0]{{\mathcal A}}
\DeclareRobustCommand{\calP}[0]{{\mathcal P}}
\DeclareRobustCommand{\calQ}[0]{{\mathcal Q}}
\DeclareRobustCommand{\calD}[0]{{\mathcal D}}
\DeclareRobustCommand{\calS}[0]{{\workers}}
\DeclareRobustCommand{\calK}[0]{{\keyspace}}
\newcommand{\code}[1]{{\textsc #1}}
\newcommand{\indic}{\mathbb{I}\,}

We proceed to analyze the conditions under which \pkgs achieves good load balance.
Recall from Section~\ref{sec:preliminaries} that we have a set $\calS$ of $n$ workers at our disposal and receive a sequence of $m$ messages  $k_1,
       \ldots, k_m$ with values from a key universe $\calK$. Upon receiving the $i$-th message with value $k_i\in \calK$, we need to decide its placement %
among the workers; decisions are irrevocable. We assume one message arrives per unit of time.
Our goal is to minimize the eventual maximum load $L(m)$, which is the same as minimizing the imbalance $I(m)$.
A simple placement scheme such as shuffle grouping provides an imbalance of at most one, but
we would like to limit the number of workers processing each key to $d \in \naturals^+$.

\spara{Chromatic balls and bins.}
We model our problem in the framework of balls and bins processes, where keys correspond to colors, messages to colored balls, and workers to bins.
Choose
 $d$ independent hash
functions $\hash{1}, \ldots, \hash{d}\colon \calK \to [n]$ uniformly at random. %
Define the \code{greedy-$d$} scheme as follows:
at time $t$, the $t$-th ball (whose color is $k_t$) is placed on the bin with minimum current load among
$\hash{1}(k_t), \ldots, \hash{d}(k_t)$, i.e., $P_t(k_t) = \operatorname{argmin}_{i \in \{\hash{1}(k_t),\ldots,\hash{d}(k_t)\}} L_i(t)$.
Recall that with key splitting there is no need to remember the choice for the next time a ball of the same color appears.

Observe that when $d = 1$,
each ball color is assigned to a unique bin so no choice has to be made; this models hash-based key grouping.
At the other extreme, when $d \gg n \ln n$, all $n$  bins are valid choices, and we obtain shuffle
grouping.%

\spara{Key distribution.}
Finally, we assume the existence of an underlying discrete distribution $\calD$ supported on $\calK$ from which ball colors
are drawn, i.e., $k_1, \ldots, k_m$ is a sequence of $m$ independent samples from $\calD$. 
Without loss of generality, we identify the set $\calK$ of keys with
$\naturals^+$ or, if $\calK$ is finite of cardinality $\keysize = |\calK|$, with $[\keysize] = \{1, \ldots, \keysize\}$.
We assume them ordered by decreasing probability: if $p_i$ is the probability of drawing key $i$ from $\calD$, then
$p_1 \ge p_2 \ge p_3 \ldots$ and $\sum_{i \in \calK} p_i = 1$. We also identify the set $\calS$ of bins with~$[n]$.

\subsection{Imbalance with \pkg}

\spara{Comparison with standard problems.}
As long as we keep getting balls of different colors, our process is identical to the standard \code{greedy-$d$} process of \citet{azar1999balanced-allocations}.
This occurs with high probability provided that %
$m$ is small enough.
But for sufficiently large $m$ (e.g., when $m \ge \frac{1}{p_1}$), repeated keys will start to arrive.
Recall that for any number of choices $d\ge 2$, the maximum imbalance after throwing $m$ balls \emph{of different colors} into $n$ bins with the standard \code{greedy-$d$} process is $\frac{\ln \ln n}{\ln d} + \frac{m}{n} + O(1)$. %
Unfortunately, such strong bounds (independent of $m$) cannot apply to our setting.
To gain some intuition on what may go wrong, consider the following examples where $d$$=$$2$.

Note that for the maximum load not to be much larger than the average load, the number of bins used must not exceed $O(1/p_1)$, where $p_1$ is the maximum key probability.
Indeed, at any time we expect
the two bins $h_1(1), h_2(1)$ to contain together at
least a $p_1$ fraction of all balls, just counting the occurrences of a single key. Hence the expected maximum load among the two grows at a rate
of at least $p_1/2$ per unit of
time, while the overall average load increases by exactly $\frac{1}{n}$ per unit of time. Thus, if $p_1 > 2/ n$, the expected imbalance at time $m$ will be lower bounded by
$(\frac{p_1}2 - \frac1n) m$, which grows \emph{linearly} with $m$. This holds irrespective of the placement scheme used.

However, requiring $p_1 \le 2 / n$ is not enough to prevent imbalance $\Omega(m)$. Consider the uniform distribution
over $n$ keys. Let $B = \bigcup_{i \le n} \{ \hash{1}(i), \hash{2}(i) \}$ be the set of all bins that belong to one of the potential choices for some key. 
As is well-known, the expected size of $B$ is $n - n \left(1 - \frac{1}{n}\right)^{2n} \approx n (1 -
    \frac{1}{e^2})$. So all $n$ keys use only an $(1-\frac{1}{e^2}) \approx 0.865$ fraction of all
bins, and roughly $0.135 n$ bins will remain unused. In fact the imbalance after $m$ balls
will be at least $\frac{m}{0.865 n}-\frac{m}{n} \approx 0.156 m$.
The problem is that most concrete instantiations of our two random hash functions cause the existence of
an ``overpopulated'' set $B$ of bins inside which
the average bin load must grow faster than the average load across all bins.
(In fact, this case subsumes our first example above, where $B$ was $\{ \hash{1}(1), \hash{2}(1) \}$.)

Finally, even in the absence of overpopulated bin subsets, some inherent imbalance
is     due to deviations between the empirical and true key distributions.
For instance, suppose there are two
keys $1, 2$ with equal probability $\frac12$ and $n =4$ bins. With constant probability, key 1 is assigned to bins $1, 2$
and key 2 to bins $3, 4$. This situation looks perfect because the \code{greedy-2} choice will send each
occurrence of key 1 to bins $1,2$ alternately so the loads of bins $1, 2$ will always equal up to $\pm 1$.
However, the number of balls with key 1 seen is likely to deviate from $m/2$ by roughly $\Theta(\sqrt m)$, so either the top two or the bottom two bins will receive
$m / 4 + \Omega(\sqrt m)$ balls, and the imbalance will be $\Omega(\sqrt m)$ with constant probability.

In the remainder of this section we carry out our analysis, which broadly construed asserts that the above are the only impediments to achieve good balance. 

\spara{Statement of results.}
We  noted  that once the number of bins exceeds $2/p_1$ (where $p_1$ is the maximum key
frequency), the maximum load will be dominated by the loads of the bins to which the most frequent key is mapped. 
Hence the main case of interest is where $p_1 = O(\frac{1}{n})$.%

We focus on the case where the number of balls is large compared to the number of bins. 
The following results show that partial key grouping can significantly reduce the maximum load (and the imbalance), compared to key
grouping.
\begin{theorem}\label{thm:main}
Suppose we use $n$ bins and let $m \ge n^2$.
Assume a key distribution $\calD$ with maximum probability $p_1 \le \frac{1}{5 n}$.
Then
    the imbalance after $m$ steps of the \code{Greedy-$d$} process satisfies, with probability at least $1-\frac{1}{n}$,
$$ I(m) = 
\begin{cases}
O\left(\frac{m}{n} \cdot \frac{\ln n}{\ln \ln n}\right), & \text{if } d = 1\\
O\left(\frac{m}{n}\right), & \text{if } d \geq 2\\
    \end{cases}
    . $$
\mycomment{
the following hold:
\begin{itemize}
\item For any $d\ge 2$, the imbalance after $m$ steps of the \code{Greedy-$d$} process satisfies, with probability at least $1-\frac{1}{n}$,
$$ I(m) = \Theta\left(\frac{m}{n}\right) . $$
Moreover, with probability at least $1-\frac{1}{n}$ over the choice of the $d$ hash functions, any load balancing policy using the same set of
choices for each key must have imbalance
$\Omega\left(\frac{m}{n}\right)$.
\item For some distributions $\calD$,
the imbalance after $m$ steps of the \code{Greedy-1} process satisfies, with probability at least $1-\frac{1}{n}$,
$$ I(m) = \Theta\left(\frac{m}{n} \cdot \frac{\ln n}{\ln \ln n}\right) . $$
\end{itemize}
}
\end{theorem}

As the next result shows, the bounds above are best-possible.\footnote{\scriptsize However, the imbalance can be much smaller than the worst-case bounds from Theorem~\ref{thm:main} if the probability of most keys is much smaller than $p_1$, which is the case in many setups. }
\begin{theorem}\label{thm:lb}
There is a distribution $\calD$ satisfying the hypothesis of Theorem~\ref{thm:main} such that
the imbalance after $m$ steps of the \code{Greedy-$d$} process satisfies, with probability at least $1-\frac{1}{n}$,
$$ I(m) = 
\begin{cases}
\Omega\left(\frac{m}{n} \cdot \frac{\ln n}{\ln \ln n}\right), & \text{if } d = 1\\
\Omega\left(\frac{m}{n}\right), & \text{if } d \geq 2\\
    \end{cases}
    . $$
\end{theorem}

We omit the proof of Theorem~\ref{thm:lb} (it follows by considering a uniform distribution over $5 n$ keys).
The next section is devoted to the proof of the upper bound, Theorem~\ref{thm:main}. 

\subsection{Proof}
\mycomment{
First define $\tau_d^\calD$.

\begin{theorem}
Conditioned on $\tau_d^\calD$, with high probability after $m$ throws the normalized imbalance after running the \code{Greedy-$d$} process satisfies
$$
\widehat{I}(m) \le (\tau_d^\calD - 1) + \sqrt{\frac{n \ln n \ln m}{m}}.
$$
Also, the optimal normalized imbalance after the hash functions are chosen satisfies
$$\widehat{I}^*(m) \ge (\tau_d^\calD - 1) + \sqrt{\frac{\ln n}{m}}.$$
\end{theorem}

\begin{theorem}
With high probability,
     $$\tau_d^\calD = \min(O(\frac{\ln n}{\ln \ln n}), \frac{1}{n} + \sqrt{p_1 n \ln n}$$
and
        for $d \ge 2$,
        $$\tau_d^\calD = \min(O(\frac{1}{n} + p_1)).$$
\end{theorem}

In particular, when $p_1 \le \frac{1}{n^3 \ln n}$, the maximum load will be $O(m / n)$ whether we use one choice or $d > 1$ choices.
On the other hand,
From these results we see that when $p_1 \ge \frac{1}{n}$, the imbalance with $d = 1$ or
In the typical case we consider where $p_1 < \frac{1}{n}$,

}

\spara{Concentration inequalities.}
We recall the following results, which we need to prove our main theorem.
\begin{theorem}[Chernoff bounds]\label{chernoff}
Suppose $\{X_i\}$ is a finite sequence of independent random variables with $X_i \in [0, M]$ and let $Y = \sum_i X_i$, $\mu = \sum_i \expect[X_i]$.
Then for all $\beta \ge \mu$,
$$ \Pr[ Y \ge \beta ] \le C(\mu, \beta, M) , $$
where
$$ C(\mu, \beta, M) \triangleq \exp\Big(-\frac{\beta \ln(\frac{\beta}{e \mu}) + \mu}{M}\Big). $$%
\end{theorem}

\begin{theorem}[McDiarmid's inequality]\label{bounded_dif}
Let $X_1, \ldots, X_n$ be a vector of independent random variables and let $f$ be a function satisfying
    $ |f(a) - f(a') | \le 1$
whenever the vectors $a$ and $a'$ differ in just one coordinate. Then
$$ \Pr[  f(X_1, \ldots, X_n) > \expect[ f(X_1, \ldots, X_n) ] + \lambda ] \le \exp(-2 \lambda^2).$$
\end{theorem}

\spara{The $\mu_r$ measure of bin subsets.}
For every nonempty set of bins $S \subseteq[n]$ and $1 \le r \le d$, define
$$ \mu_r(S) = \sum\{ p_i \mid \{\hash{1}(i), \ldots, \hash{r}(i)\} \subseteq B \} .$$
We will be  interested in $\mu_1(B)$ (which measures the probability that a random key from $\calD$ will have its choice inside $B$) and $\mu_d(B)$
(which measures the probability that a random key from $\calD$ will have all its choices inside $B$).
Note that $\mu_1(B) = \sum_{j \in B} \mu_1(\{j\})$ and $\mu_d(B) \le \mu_1(B)$.

\begin{lemma}\label{lem:mu1}
For every $B \subseteq [n]$,
     $\expect [\mu_1(B)]= \frac{|B|}{n}$
and, if $p_1 \le \frac{1}{n}$,
       $$ \Pr \left[ \mu_1(B) \ge \frac{|B|}{n} (e \lambda) \right]  \le \left(\frac{1}{\lambda^\lambda}\right)^{|B|}. $$
\end{lemma}
\begin{proof}
The first claim follows from linearity of expectation and the fact that $\sum_i p_i = 1$.
For the second, let $|B| = k$. Using Theorem~\ref{chernoff},
$\Pr \left[ \mu_1(B) \ge \frac{k}{n} (e \lambda) \right]$ is at most
    $$
C\left(\frac{k}{n}, \frac{k}{n} e \lambda, p_1\right)
                                                   \le \exp\left(-\frac{k}{n p} e \lambda \ln \lambda \right) 
                                                   \le \exp(-k \lambda \ln \lambda),$$

since $n p_1 \le 1$.
\end{proof}

\mycomment{
    \begin{corollary}
    With high probability, for all $B \subseteq [n]$ it holds that
    $$ \mu_1(B) \le O\left(\frac{|B|}{n} \frac{\ln n}{\ln \ln n}\right). $$
    \end{corollary}
    \begin{proof}
    Since $\mu_1(B) = \sum_{j \in B} \mu_1(\{j\})$, it suffices to show the claim when $|B| = 1$.
    This follows from Lemma~\ref{lem:mu1} by setting $|B| = 1$, $\lambda = O(\frac{\ln n}{\ln \ln n})$, and applying the union bound.
    \end{proof}
}

\begin{lemma}\label{lem:mu2}
For every $B \subseteq [n]$,
     $\expect [\mu_d(B)]= \left(\frac{|B|}{n}\right)^d$
and, provided that $p_1 \le \frac{1}{5 n}$,
       $$ \Pr \left[ \mu_d(B) \ge \frac{|B|}{n} \right]  \le \left(\frac{e |B|}{n}\right)^{5 |B|}. $$
\end{lemma}
\begin{proof}
Again the first claim is easy. For the second, let $|B| = k$. Using Theorem~\ref{chernoff},
$\Pr \left[ \mu_d(B) \ge \frac{k}{n} \right]$ is at most
\begin{align*}
 C\left(\Big(\frac{k}{n}\Big)^d, \frac{k}{n}, p_1\right)
                                            &\le \exp\left(-\frac{k (d - 1)}{n p_1} \ln \left(\frac{n}{e k}\right) \right) \\
                                            &\le \exp\left(-5 k \ln \left(\frac{n}{e k}\right)\right)
\end{align*}    

since $n p_1 \le \frac{1}{5}$.
\end{proof}

\begin{corollary}\label{coro:mu2}
Assume $p_1 \le \frac{1}{4n}$, $d \ge 2$.  Then, with high probability, 
       $$ \max \left\{ \frac{\mu_d(B)}{|B| / n} \middle| B \subseteq [n], |B| \le \frac{n}{5} \right\} \le 1.$$
\end{corollary}
\begin{proof}
We use Lemma~\ref{lem:mu1} and the union bound. The probability that the claim fails to hold is bounded by
\begin{align*}
 \sum_{|B| \le n / 5} \Pr \left[ \mu_d(B) \ge \frac{k}{n} \right] &\le \sum_{k \le n/5} \binom{n}{k} \left(\frac{e k}{n}\right)^{5k} \\
                                                                  &\le \sum_{k \le n/5} \left(\frac{e n}{k}\right)^k \left(\frac{e k}{n}\right)^{5k}  
                                                                  &=o\left(\frac{1}{n}\right),
\end{align*}                                                                  
where we used  $\binom{n}{k} \le \left(\frac{e n}{k}\right)^k$, valid for all $k$.
\end{proof}

For a scheduling algorithm $\calA$ and a set $B \subseteq [n]$ of bins, write $L^\calA_B(t) = 
\max_{j\in B} L_j(t)$ for the maximum load among the bins in $B$ after $t$ balls have been
processed by $\calA$.
\begin{lemma}\label{lem:perfect}
Suppose there is a set $A \subseteq [n]$ of bins such that for all $T \subseteq A$,
        $ \mu_d(T) \le \frac{|T|}{n}. $
Then $\calA=\text{\code{Greedy-$d$}}$ satisfies $L^\calA_A(m) = O(\frac{m}{n}) + L^\calA_{[n]\setminus A}(m)$ with high probability.
\end{lemma}
\begin{proof}
We use a coupling argument. Consider the following two independent processes $\calP$ and $\calQ$: $\calP$ proceeds as
\code{Greedy-$d$}, while $\calQ$ picks the bin for each ball independently at
random from $[n]$ and increases its load. Consider any time~$t$ at which the load vector is $\omega_t \in \naturals^n$
and $M_t=M(\omega_t)$ is the set of bins with
maximum load. After handling the $t$-th ball, let $X_t$ denote the event that $\calP$ increases the maximum load in $A$ because the new ball has all choices in
$M_t \cap A$, and $Y_t$ denote the event that $\calQ$ increases the maximum load in $A$. Finally, let $Z_t$ denote the event
that $\calP$ increases the maximum load in $A$ because the new ball has some choice in $M_t \cap A$ and some choice in $M_t \setminus
A$, but the load of one of  its choices in $M_t \cap A$ is no larger. We identify these events with their indicator
random variables.

Note that the maximum load in $A$ at the end of Process $\calP$ is $L_A^\calP(m) = \sum_{t\in[m]} (X_t + Z_t)$, while  at the
end of Process $\calQ$ is $L_A^\calQ(m) = \sum_{t\in[m]}
Y_t$. Conditioned on any load vector $\omega_t$, the probability of $X_t$ is $$\Pr[ X_t \mid \omega_t] =
\mu_d(M_t \cap A) \le \frac{|M_t \cap A|}{n}  \le \frac{|M_t|}{n} = \Pr[ Y_t \mid \omega_t ],$$
    So $\Pr [ X_t   \mid
\omega_t ] \le \Pr [Y_t \mid \omega_t]$, which implies that for any $b \in \naturals$, $\Pr [ \sum_{t \in [m]} X_t
\le b ] \ge \Pr[
\sum_{t \in [m]} Y_t
\le b ].$ But with high probability, the maximum load of Process $\calQ$ is $b = O(m / n)$, so $\sum_t X_t = O(m / n)$
holds with at least the same probability. On the other hand, $\sum_t Z_t \le L_{[n]\setminus A}^\calP(m)$ because each occurrence of
$Z_t$ increases the maximum load on $A$, and once a time $t$ is reached such that $L_A^\calP(t) > L_{[n]\setminus
    A}^\calP(m)$, event $Z_t$ must cease to happen. Therefore
$L_A^\calP(m) = \sum_{t\in[m]} X_t + \sum_{t\in[m]} Z_t \le O(m / n) + L_{[n]\setminus A}^\calP(m)$,
yielding the result.
\end{proof}

\begin{proof}[Proof of Theorem~\ref{thm:main}]
Let $$A = \left\{j \in [n] \mid \mu_1(\{j\}) \ge \frac{3 e}{n}\right\}.$$
Observe that every bin $j \notin A$ has $\mu_1(\{j\}) < \frac{3 e}{n}$ and this implies that, conditioned on any choice of hash functions, the maximum load of all bins outside $A$ is
at most $\frac{20 m}{n}$ with high probability.\footnote{This is by majorization with the process that just throws every ball to the \emph{first}
    choice; see, e.g,~\citet{azar1999balanced-allocations}.} Therefore our task reduces to showing that the maximum load of the bins in $A$ is $O(\frac{m}{n})$.

Consider the sequence $X_1, \ldots, X_\keysize$ of random variables given by $X_i = \hash{1}(i)$, and let $f(X_1, X_2,
        \ldots, X_\keysize) = |A|$ denote the number of bins $j$ with $\mu_1(\{j\}) \ge \frac{3 e}{n}$.
By Lemma~\ref{lem:mu1}, 
$ \expect[|A|] = \expect[ f ] \le \frac{1}{27} $.
Moreover, the function $f$
satisfies the hypothesis of Theorem~\ref{bounded_dif}.%
We conclude that, with high probability, $|A| \le \frac{n}{5}.$

Now assume that the thesis of Corollary~\ref{coro:mu2} holds, which happens except with probability $o(1/n)$.
Then we have that for all $B \subseteq A$, $\mu_d(B) \le \frac{|B|}{n}.$
Thus Lemma~\ref{lem:perfect} applies to $A$. This means that after throwing $m$ balls, the maximum load among the bins in $A$ is $O(\frac{m}{n})$, as
we wished to show.
\end{proof}

\section{Evaluation}
\label{sec:evaluation}

We assess the performance of our proposal by using both simulations and a real deployment.
In so doing, we answer the following questions:
\begin{squishlist}
\item[\textbf{Q1:}]
What is the effect of key splitting on \potc?
\item[\textbf{Q2:}]
How does local estimation compare to a global oracle?
\item[\textbf{Q3:}]
How robust is \pkg? %
\item[\textbf{Q4:}]
What is the overall effect of \pkg on applications deployed on a real \dspe?
\end{squishlist}

\subsection{Experimental Setup}

\begin{table}[t]
\caption{Summary of the datasets used in the experiments: number of messages, number of keys and percentage of messages having the most frequent key ($p_1$).} %
\centering
\small
\begin{tabular}{l c r r r}
\toprule
Dataset		&	Symbol	&	Messages				&	Keys			& 	$p_1$(\%)		\\
\midrule
Wikipedia		&	WP		&	\num{22}M		&	\num{2,9}M		&	\num{9,32}	\\
Twitter		&	TW		&	\num{1,2}G		&	\num{31}M		&	\num{2,67}	\\
Cashtags		&	CT		&	\num{690}k		&	\num{2,9}k		&	\num{3,29}	\\
\midrule
Synthetic 1 	&	LN$_1$	&	\num{10}M		&	\num{16}k			&	\num{14,71}	\\
Synthetic 2 	&	LN$_2$	&	\num{10}M		&	\num{1,1}k		&	\num{7,01}	\\
\midrule
LiveJournal	&	LJ		&	\num{69}M		&	\num{4,9}M		&	\num{0,29}	\\
Slashdot0811	&	SL$_1$	&	\num{905}k		&	\num{77}k			&	\num{3,28}	\\
Slashdot0902 	&	SL$_2$	&	\num{948}k		&	\num{82}k			&	\num{3,11}	\\
\bottomrule
\end{tabular}
\label{tab:summary-datasets}
\end{table}

\spara{Datasets.}
Table~\ref{tab:summary-datasets} summarizes the datasets used.
We use two main real datasets, one from \emph{Wikipedia} and one from \emph{Twitter}.
These datasets were chosen for their large size, their different degree of skewness, and because they are representative of Web and online social network domains.
The Wikipedia dataset (WP)\footnote{\url{http://www.wikibench.eu/?page_id=60}} is a log of the pages visited during a day in January 2008. %
Each visit is a message and the page's URL represents its key.
The Twitter dataset (TW) is a sample of tweets crawled during July 2012.
Each tweet is parsed and split into its words, which are used as the key for the message.

An additional Twitter dataset comprises a sample of tweets crawled in November 2013.
The keys for the messages are the \emph{cashtags} in these tweets.
A \emph{cashtag} is a ticker symbol used in the stock market to identify a publicly traded company preceded by the dollar sign (e.g., \$AAPL for Apple).
Popular cash tags change from week to week.
This dataset allows to study the effect of shift of skew in the key distribution.

We also generate two synthetic datasets (LN$_1$, LN$_2$) with keys following a log-normal distribution, a commonly used heavy-tailed skewed distribution~\citep{talwar2007weightedcase}.
The parameters of the distribution ($\mu_1$=$1.789$, $\sigma_1$=$2.366$; $\mu_2$=$2.245$, $\sigma_2$=$1.133$) come from an analysis of Orkut, and try to emulate workloads from the online social network domain~\cite{benevenuto09characterizingusers}.

Finally, we experiment on three additional datasets comprised of directed graphs\footnote{\url{http://snap.stanford.edu/data}} (LJ, SL$_1$, SL$_2$).
We use the edges in the graph as messages and the vertices as keys. %
These datasets are used to test the robustness of \pkgs to skew in partitioning the stream \emph{at the sources}, as explained next.
They also represent a different kind of application domain: streaming graph mining.

\spara{Simulation.}
We process the datasets by simulating the \dagr presented in Figure~\ref{fig:imbalance}.
The stream is composed of timestamped keys that are read by multiple independent sources (\sources) via shuffle grouping, unless otherwise specified.
The sources forward the received keys to the workers (\workers) downstream.
In our simulations we assume that the sources perform data extraction and transformation, while the workers perform data aggregation, which is the most computationally expensive part of the \dagr.
Thus, the workers are the bottleneck in the \dagr and the focus for the load balancing.

\subsection{Experimental Results}

\begin{figure*}[t]
\begin{center}
	\includegraphics[width=\textwidth]{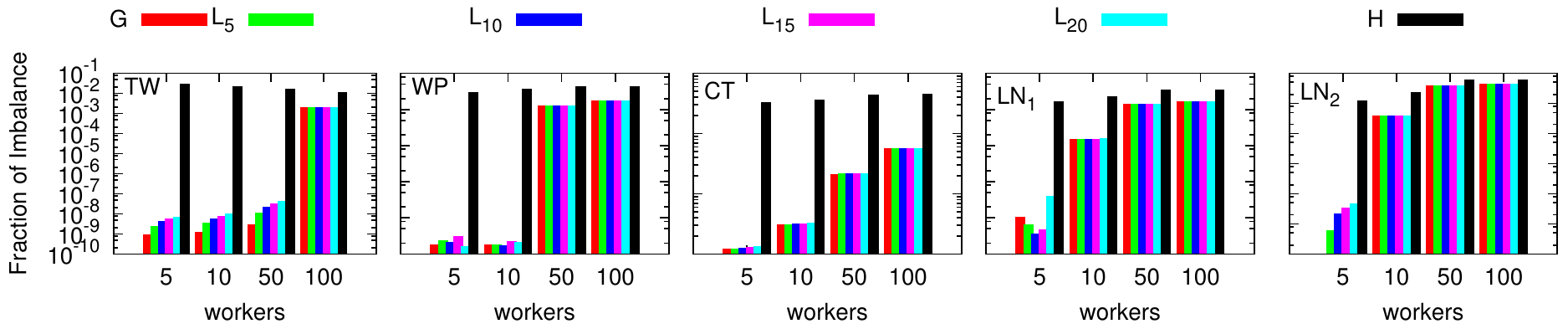}
	\caption{Fraction of average imbalance with respect to total number of messages for each dataset, for different number of workers and number of sources.} %
\label{fig:avg-GLH}
\end{center}
\end{figure*}

\spara{Q1.}
We measure the imbalance in the simulations when using the following techniques:
\begin{squishlist}
\item[H:] Hashing, which represents standard key grouping (\kg) and is our main baseline.
We use a $64$-bit Murmur hash function to minimize the probability of collision.
\item[\potc:] Power of two choices \emph{without} using key splitting, i.e., traditional \potc applied to key grouping.
\item[On-Greedy:] Online greedy algorithm that picks the least loaded worker to handle a new key.
\item[Off-Greedy:] Offline greedy sorts the keys by decreasing frequency and executes On-Greedy.
\item[\pkgs:] \potc with key splitting. %
\end{squishlist}

\begin{table}[t]
\caption{Average imbalance when varying the number of workers for the Wikipedia and Twitter datasets.}
\tabcolsep=0.1cm
\centering
\small
\begin{tabular}{l r r r r r r r r}
\toprule
Dataset		&	\multicolumn{4}{c}{WP} 					&	\multicolumn{4}{c}{TW}					\\
\cmidrule(lr){2-5} \cmidrule(lr){6-9}
$W$			&	5		&	10		&	50		&	100		&	5		&	10		&	50		&	100		\\
\midrule
\pkgs		&	0.8		&	2.9		&	5.9e5	&	8.0e5	&	0.4		&	1.7		&	2.74		&	4.0e6	\\
Off-Greedy	&	0.8		&	0.9		&	1.6e6	&	1.8e6	&	0.4		&	0.7		&	7.8e6	&	2.0e7	\\
On-Greedy	&	7.8		&	1.4e5	&	1.6e6	&	1.8e6	&	8.4		&	92.7		&	1.2e7	&	2.0e7	\\
\potc			&	15.8		&	1.7e5	&	1.6e6	&	1.8e6	&	2.2e4	&	5.1e3	&	1.4e7	&	2.0e7	\\
Hashing		&	1.4e6	&	1.7e6	&	2.0e6	&	2.0e6	&	4.1e7	&	3.7e7	&	2.4e7	&	3.3e7	\\
\bottomrule
\end{tabular}
\label{tab:key-splitting}
\end{table}

Note that \pkgs is the only method that uses key splitting.
Off-Greedy knows the whole distribution of keys so it represents an unfair comparison for online algorithms.

Table~\ref{tab:key-splitting} shows the results of the comparison on the two main datasets WP and TW.
Each value is the average imbalance measured throughout the simulation.
As expected, hashing performs the worst, creating a large imbalance in all cases.
While \potc performs better than hashing in all the experiments, it is outclassed by On-Greedy on TW.
On-Greedy performs very close to Off-Greedy, which is a good result considering that it is an online algorithm.
Interestingly, \pkgs performs even better than Off-Greedy.
Relaxing the constraint of \kg allows to achieve a load balance comparable to offline algorithms.

We conclude that \potc alone is not enough to guarantee good load balance, and key splitting is fundamental not only to make the technique practical in a distributed system, but also to make it effective in a streaming setting.
As expected, increasing the number of workers also increases the average imbalance.
The behavior of the system is binary: either well balanced or largely imbalanced.
The transition between the two states happens when the number of workers surpasses the limit $O(1/p_1)$ described in Section~\ref{sec:theory}, which happens around $50$ workers for WP and $100$ for TW.

\spara{Q2.}
Given the aforementioned results, we focus our attention on \pkgs henceforth.
So far, it still uses global information about the load of the workers when deciding which choice to make.
Next, we experiment with local estimation, i.e., each source performs its own estimation of the worker load, based on the sub-stream processed so far.

We consider the following alternatives:
\begin{squishlist}
\item[G:] \pkgs with global information of worker load.
\item[L:] \pkgs with local estimation of worker load and different number of sources, e.g., L$_5$ denotes $\numsources=5$.
\item[LP:] \pkgs with local estimation and \emph{periodic probing} of worker load every $T_p$ minutes. For instance, L$_5$P$_1$ denotes $\numsources=5$ and $T_p=1$.
When probing is executed, the local estimate vector is set to the actual load of the workers.
\end{squishlist}

Figure~\ref{fig:avg-GLH} shows the average imbalance (normalized to the size of the dataset) with different techniques, for different number of sources and workers, and for several datasets.
The baseline (H) always imposes very high load imbalance on the workers.
Conversely, \pkgs with local estimation (L) has always a lower imbalance.
Furthermore, the difference from the global variant (G) is always less than one order of magnitude.
Finally, this result is robust to changes in the number of sources.

Figure~\ref{fig:avg-time} displays the imbalance of the system through time $I(t)$ for TW, WP and CT, 5 sources, and for $W=10$ and $50$.
Results for $W=5$ and $W=100$ are omitted as they are similar to $W=10$ and $W=50$, respectively.
\pkgs with global information (G) and its variant with local estimation (L$_5$) perform best.
Interestingly, even though both G and L achieve very good load balance, their choices are quite different.
In an experiment measuring the agreement on the destination of each message, G and L have only 47\% Jaccard overlap.
Hence, L reaches a local minimum which is very close in value to the one obtained by G, although different.
Also in this case, good balance can only be achieved up to a number of workers that depends on the dataset.
When that number is exceeded, the imbalance increases rapidly, as seen in the cases of WP and partially for CT for $W=50$, where all techniques lead to the same high load imbalance.

Finally, we compare our local estimation strategy with a variant that makes use of periodic probing of workers' load every minute (L$_5$P$_1$).
Probing removes any inconsistency in the load estimates that the sources may have accumulated.
However, interestingly, this technique does not improve the load balance, as shown in Figure~\ref{fig:avg-time}.
Even increasing the frequency of probing does not reduce imbalance (results not shown in the figure for clarity).
In conclusion, local information is sufficient to obtain good load balance, therefore it is not necessary to incur the overhead of probing.

\begin{figure}[t]
\begin{center}
	\includegraphics[scale=0.74]{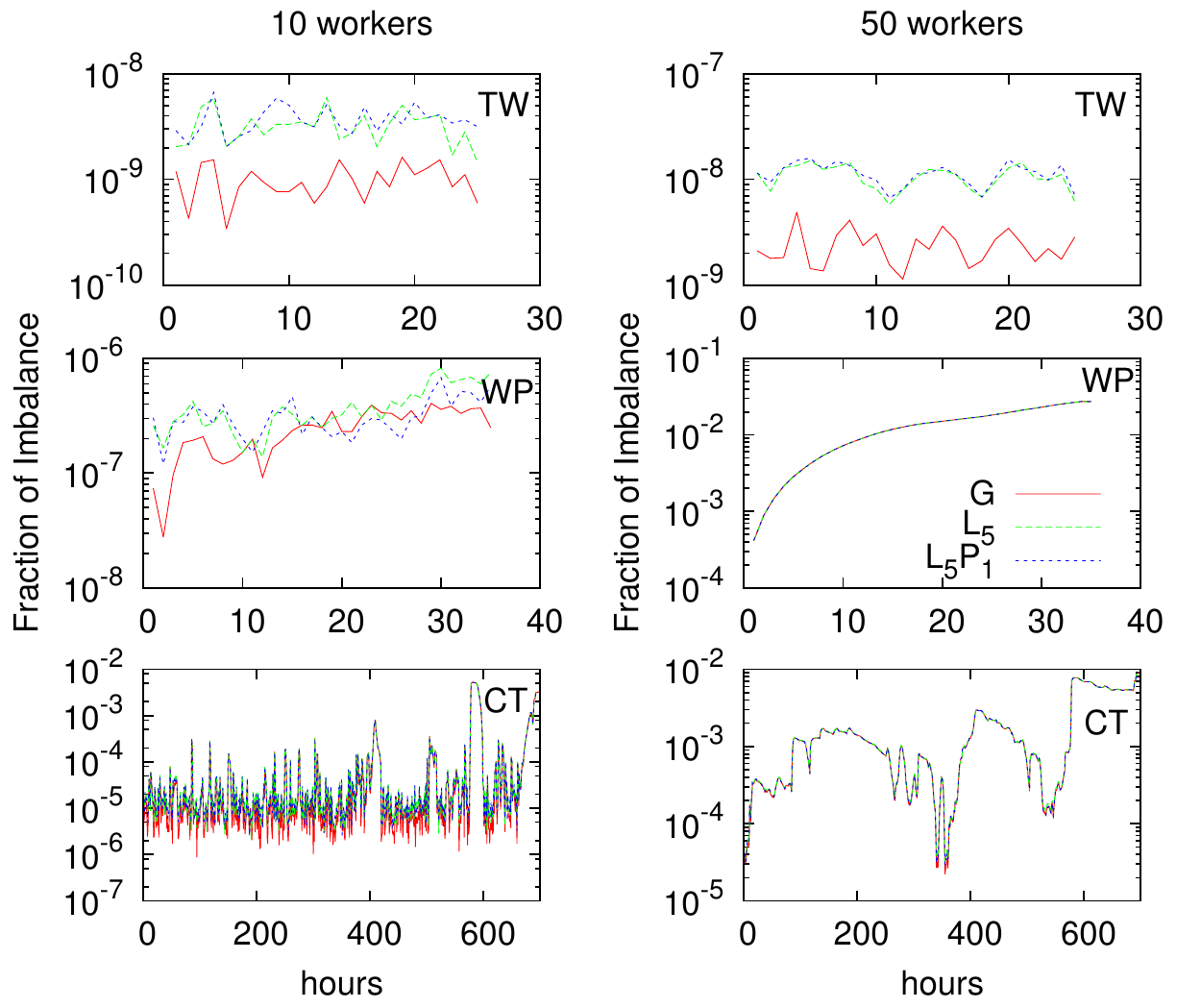}
	\caption{Fraction of imbalance through time for different datasets, techniques, and number of workers, with $S=5$.}
	\label{fig:avg-time}
\hspace{-4mm}
\end{center}
\end{figure}

\enlargethispage{\baselineskip}

\spara{Q3.}
To operationalize this question, we use the directed graphs datasets.
We use \kg to distribute the messages to the sources to test the robustness of \pkgs to \emph{skew in the sources}, i.e., when each source forwards an uneven part of the stream.
We simulate a simple application that computes a function of the incoming edges of a vertex (e.g., in-degree, PageRank).
The input keys for the source \pe is the source vertex id, while the key sent to the worker \pe is the destination vertex id, that is, the source \pe inverts the edge.
This schema projects the out-degree distribution of the graph on sources, and the in-degree distribution on workers, both of which are highly skewed.

\begin{figure}[t]
\begin{center}
	\includegraphics[width=\columnwidth]{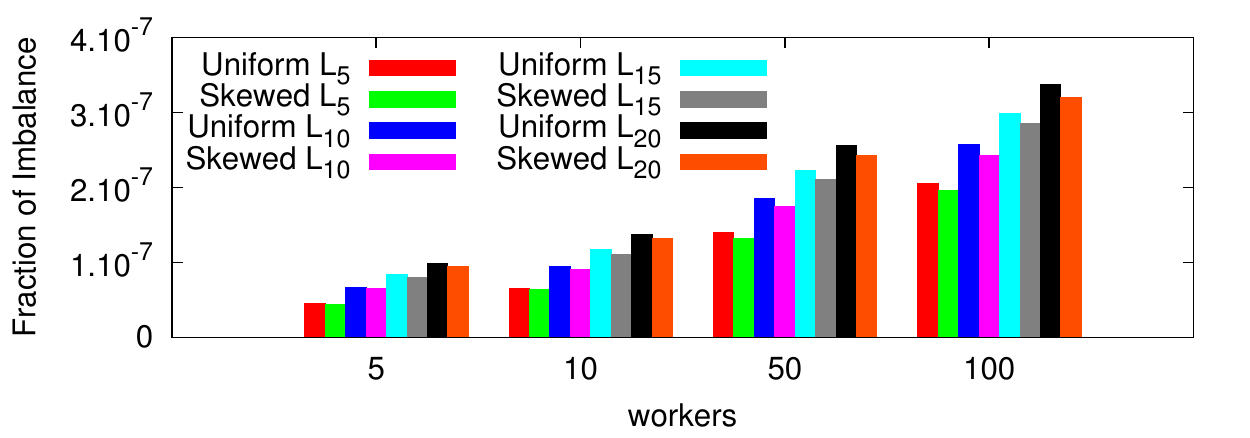}
	\caption{Fraction of average imbalance with uniform and skewed splitting of the input keys on the sources when using the LJ graph.}
	\label{fig:avg-skewed-shuffled}
\end{center}
\end{figure}

Figure~\ref{fig:avg-skewed-shuffled} shows the average imbalance for the experiments with a skewed split of the keys to sources for the LJ social graph (results on SL$_1$ and SL$_2$ are similar to LJ and are omitted due to space constraint).
For comparison, we include the results when the split is performed \emph{uniformly} using shuffle grouping of keys on sources.
On average, the imbalance generated by the skew on sources is similar to the one obtained with uniform splitting.
As expected, the imbalance slightly increases as the number of sources and workers increase, but, 
in general, it remains at very low absolute values.

To answer Q3, we additionally experiment with drift in the skew distribution by using the cashtag dataset (CT).
The bottom row of Figure~\ref{fig:avg-time} demonstrates that all techniques achieve a low imbalance, even though the change of key popularity through time generates occasional spikes.

In conclusion, \pkgs is robust to skew on the sources, and can therefore be chained to key grouping.
It is also robust to the drift in key distribution common of many real-world streams.

\enlargethispage{\baselineskip}

\spara{Q4.}
We implement and test our technique on the streaming top-k word count example, and perform two experiments to compare \pkgs, \kg, and \sg on WP.
We choose word count as it is one of the simplest possible examples, thus limiting the number of confounding factors.
It is also representative of many data mining algorithms as the ones described in Section~\ref{sec:applications} (e.g., counting frequent items or co-occurrences of feature-class pairs).
Due to the requirement of real-world deployment on a \dspe, we ignore techniques that require coordination (i.e., \potc and On-Greedy).
We use a topology configuration of a single source along with 9 workers (counters) running on a storm cluster of 10 virtual servers. %
We report overall throughput, end-to-end latency, and memory usage. %

In the first experiment, we emulate different levels of CPU consumption per key by adding a fixed delay to the processing.
We prefer this solution over implementing a specific application in order to be able to control the load on the workers.
We choose a range that is able to bring our configuration to a saturation point, although the raw numbers would vary for different setups.
Even though real deployments rarely operate at saturation point, \pkgs allows better resource utilization, therefore supporting the same workload on a smaller number of machines.
In this case, the minimum delay ($0.1$ms) corresponds approximately to reading 400kB sequentially from memory, while the maximum delay ($1$ms) to $\frac{1}{10}$-th of a disk seek.\footnote{\url{http://brenocon.com/dean_perf.html}}
Nevertheless, even more expensive tasks exist: parsing a sentence with NLP tools can take up to $500$ms.\footnote{\url{http://nlp.stanford.edu/software/parser-faq.shtml\#n}}

The system does not perform aggregation in this setup, as we are only interested in the raw effect on the workers.
Figure~\ref{fig:throughput_memory_cpu_aggr}(a) shows the throughput achieved when varying the CPU delay for the three partitioning strategies.
Regardless of the delay, \sg and \pkgs perform similarly, and their throughput is higher than \kg.
The throughput of \kg is reduced by $\approx 60\%$ when the CPU delay increases tenfold, while the impact on \pkgs and \sg is smaller ($\approx 37\%$ decrease).
We deduce that reducing the imbalance is critical for clusters operating close to their saturation point, and that \pkgs is able to handle bottlenecks similarly to \sg and better than \kg.
In addition, the imbalance generated by \kg translates into longer latencies for the application.
When the workers are heavily loaded, the average latency with \kg is up to $45\%$ larger than with \pkgs. %
Finally, the benefits of \pkgs over \sg regarding memory are substantial.
Overall, \pkgs ($3.6M$ counters) requires about $30\%$ more memory than \kg ($2.9M$ counters), but about half the memory of \sg ($7.2M$ counters).

\begin{figure}[t]
\begin{center}
	\includegraphics[scale=0.7]{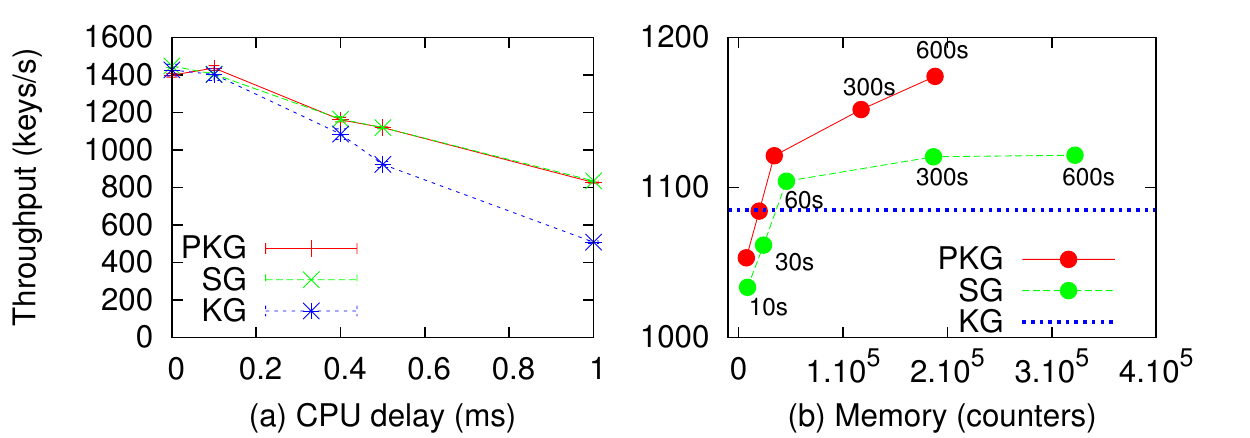}
	\caption{(a) Throughput for \pkgs, \sg and \kg for different CPU delays. (b) Throughput for \pkgs and \sg vs. average memory for different aggregation periods.}
	\label{fig:throughput_memory_cpu_aggr}
\end{center}
\end{figure}

In the second experiment, we fix the CPU delay to $0.4$ms per key, as it is the saturation point for \kg in our setup.
We activate the aggregation of counters at different time intervals $T$ to emulate different application policies for when to receive up-to-date top-k word counts.
In this case, \pkgs and \sg need additional memory compared to \kg to keep partial counters.
Shorter aggregation periods reduce the memory requirements, as partial counters are flushed often, at the cost of a higher number of aggregation messages.
Figure~\ref{fig:throughput_memory_cpu_aggr}(b) shows the relationship between throughput and memory overhead for \pkgs and \sg.
The throughput of \kg is shown for comparison.
For all values of aggregation period, \pkgs achieves higher throughput than \sg, with lower memory overhead and similar average latency per message. %
When the aggregation period is above $30$s, the benefits of \pkgs compensate its extra overhead and its overall throughput is higher than when using \kg.

\section{Applications}
\label{sec:applications}

\pkgs is a novel programming primitive for stream partitioning and not every algorithm can be expressed with it.
In general, all algorithms that use shuffle grouping can use \pkgs to reduce their memory footprint.
In addition, many algorithms expressed via key grouping can be rewritten to use \pkgs in order to get better load balancing.
In this section we provide a few such examples of common data mining algorithms, and show the advantages of \pkgs. %
Henceforth, we assume that each message contains a data point for the application, e.g., a feature vector in a high-dimensional space.

\subsection{Na\"{i}ve Bayes Classifier}
A na\"{i}ve Bayes classifier is a probabilistic model that assumes independence of features.
It estimates the probability of a class $C$ given a feature vector $X$ by using Bayes' theorem.
In practice, the classifier works by counting the frequency of co-occurrence of each feature and class values.

The simplest way to parallelize this algorithm is to spread the counters across several workers via vertical parallelism, i.e., each feature is tracked independently in parallel.
Following this design, the algorithm can be implemented by the same pattern used for the \kg example in Section~\ref{sec:existing-partitioning}. %
Sparse datasets often have a skewed distribution of features, e.g., for text classification.
Therefore, this implementation suffers from the same load imbalance, which \pkgs solves.

Horizontal parallelism can also be used to parallelize the algorithm, i.e., by shuffling messages to separate workers.
This implementation uses the same pattern as the \dagr in the \sg example in Section~\ref{sec:existing-partitioning}. %
The count for a single feature-class pair is distributed across several workers, and needs to be combined at prediction (query) time. %
This combination requires broadcasting the query to all the workers, as a feature can be tracked by any worker.
This implementation, while balancing the work better than key grouping, requires an expensive query stage that may be affected by stragglers.

\pkgs tracks each feature on two workers and avoids replicating counters on all workers.
Furthermore, the two workers are deterministically assigned for each feature.
Thus, at query time, the algorithm needs to probe only two workers for each feature, rather than having to broadcast it to all the workers.
The resulting query phase is less expensive and less sensitive to stragglers than with shuffle grouping.

\subsection{Streaming Parallel Decision Tree}

A decision tree is a classification algorithm that uses a tree-like model where nodes are tests on features, branches are possible outcomes, and leafs are class assignments.

\citet{ben-haim2010spdt} propose an algorithm to build a streaming parallel decision tree that uses approximated histograms to find the test value for continuous features.
Messages are shuffled among $W$ workers.
Each worker generates histograms independently for its sub-stream, one histogram for each feature-class-leaf triplet.
These histograms are then periodically sent to a single aggregator that merges them to get an approximated histogram for the whole stream.
The aggregator uses this final histogram to grow the model by taking split decisions for the current leaves in the tree.
Overall, the algorithm keeps $W \times D \times C \times L$ histograms, where $D$ is the number of features, $C$ is the number of classes, and $L$ is the current number of leaves.

The memory footprint of the algorithm depends on $W$, so it is impossible to fit larger models by increasing the parallelism.
Moreover, the aggregator needs to merge $W \times D \times C$ histograms each time a split decision is tried, and merging the histograms is one of the most expensive operations.

Instead, \pkgs reduces both the space complexity and aggregation cost.
If applied on the features of each message, a single feature is tracked by two workers, with an overall cost of only $2$$\times$$D$$\times$$C$$\times$$L$ histograms.
Furthermore, the aggregator needs to merge only two histograms per feature-class-leaf triplet.
This scheme allows to alleviate memory pressure by adding more workers, as the space complexity does not depend on $W$.

\subsection{Heavy Hitters and Space Saving}

The heavy hitters problem consists in finding the top-k most frequent items occurring in a stream.
The \textsc{SpaceSaving}~\citep{metwally2005spacesaving} algorithm solves this problem approximately in constant time and space.
Recently, \citet{berinde2010heavyhitters} have shown that \textsc{SpaceSaving} is space-optimal, and how to extend its guarantees to merged summaries.
This result allows for parallelized execution by merging partial summaries built independently on separate sub-streams.

In this case, the error bound on the frequency of a single item depends on a term representing the error due to the merging, plus another term which is the sum of the errors of each individual summary for a given item $i$:
$$
\mid \hat{f}_i - f_i \mid \leq \Delta_{f} + \sum_{j=i}^{W}\Delta_{j}
$$
where $f_i$ is the true frequency of item $i$ and $\hat{f}_i$ is the estimated one, each $\Delta_j$ is the error from summarizing each sub-stream, while $\Delta_{f}$ is the error from summarizing the whole stream, i.e., from merging the summaries.

Observe that the error bound depends on the parallelism level $W$.
Conversely, by using \kg, the error for an item depends only on a single summary, thus it is equivalent to the sequential case, at the expense of poor load balancing.

Using \pkgs we achieve both benefits: the load is balanced among workers, and the error for each item depends on the sum of only two error
terms, regardless of the parallelism level.
However, the individual error bounds may depend on $W$.

\section{Related Work}
\label{sec:rel-work}
Various works in the literature either extend the theoretical results from the power of two choices, or apply them to the design of large-scale systems for data processing.

\spara{Theoretical results.}\label{sec:theor_choices}
Load balancing in a \dspe can be seen as a balls-and-bins problem, where $m$ balls are to be placed in $n$ bins.
The power of two choices has been extensively researched from a theoretical point of view for balancing the load among machines~\cite{mitzenmacher2001power,mitzenmacher2001potc-survey}. %
Previous results consider each ball equivalent.
For a \dspe, this assumption holds if we map balls to messages and bins to servers.
However, if we map balls to {\em keys}, more popular keys should be consider to be heavier.
\citep{talwar2007weightedcase} tackle the case where each ball has a weight drawn independently from a fixed weight distribution $\mathcal{X}$.
They prove that, as long as $\mathcal{X}$ is ``smooth'', the expected imbalance is independent of the number of balls. %
However, the solution assumes that $\mathcal{X}$ is known beforehand, which is not the case in a streaming setting.
Thus, in our work we take the standard approach of mapping balls to messages.

Another assumption common in previous works is that there is a single source of balls. 
Existing algorithms that extend \potc to multiple sources execute several rounds of intra-source coordination before taking a decision~\cite{adler1995parallelrandomized,lenzen2011parallelrandomized, park2011multiplechoices}.
Overall, these techniques incur a significant coordination overhead, which becomes prohibitive in a \dspe that handles thousands of messages per second.

\spara{Stream processing systems.}
Existing load balancing techniques for \dspes are analogous to key grouping with rebalancing~\citep{shah2003flux,cherniack2003scalable,xing2005dynamic,gedik2013partitioning,balkesen2013adaptive,castro2013integrating}.
In our work, we consider operators that allow replication and aggregation, similar to a standard combiner in map-reduce, and show that it is sufficient to balance load among two replicas based local load estimation.
We refer to Section~\ref{sec:existing-partitioning} for a more extensive discussion of key grouping with rebalancing.
Flux monitors the load of each operator, ranks servers by load, and migrates operators from the most loaded to the least loaded server, from the second most loaded to the second least loaded, and so on~\citep{shah2003flux}.
Aurora* and Medusa propose policies to migrating operators in \dspes and federated \dspes~\citep{cherniack2003scalable}.
Borealis uses a similar approach but it also aims at reducing the correlation of load spikes among operators placed on the same server~\citep{xing2005dynamic}.
This correlation is estimated by using a finite set of load samples taken in the recent past.
\citet{gedik2013partitioning} developed a partitioning function (a hybrid between explicit mapping and consistent hashing of items to servers) for stateful data parallelism in DSPEs that leverages item frequencies to control migration cost and imbalance in the system.
Similarly, \citet{ balkesen2013adaptive} proposed frequency-aware hash-based partitioning to achieve load balance.
\citet{castro2013integrating} propose integrating common operator state management techniques for both checkpointing and migration.

\enlargethispage{\baselineskip}

\spara{Other distributed systems.}
Several storage systems use consistent hashing to allocate data items to servers~\citep{karger1997consistent}.
Consistent hashing substantially produces a random allocation and is designed to deal with systems where the set of servers available varies over time.
In this paper, we propose replicating \dspe operators on two servers selected at random.
One could use consistent hashing also to select these two replicas, using the replication technique used by Chord~\citep{stoica2001chord} and other systems.

Sparrow~\cite{ousterhout2013sparrow} is a stateless distributed job scheduler that exploits a variant of the power of two choices~\citep{park2011multiplechoices}. 
It employs batch probing, along with late binding, to assign m tasks of a job to the least loaded of $d \times m$ randomly selected workers ($d \geq 1$). 
Sparrow considers only independent tasks that can be executed by any worker.
In \dspes, a message can only be sent to the workers that are accumulating the state corresponding to the key of that message.
Furthermore, \dspes deal with messages that arrive at a much higher rate than Sparrow's fine-grained tasks, so we prefer to use local load estimation.

In the domain of graph processing, several systems have been proposed to solve the load balancing problem, e.g., Mizan~\cite{khayyat2013mizan}, GPS~\cite{salihoglu2013gps}, and xDGP~\cite{vaqueroxdgp}.
Most of these systems perform dynamic load rebalancing at runtime via vertex migration.
We have already discussed why rebalancing is impractical in our context in Section~\ref{sec:preliminaries}.

Finally, SkewTune~\cite{kwon2012skewtune} solves the problem of load balancing in MapReduce-like systems by identifying and redistributing the unprocessed data from the stragglers to other workers.
Techniques such as SkewTune are a good choice for batch processing systems, but cannot be directly applied to \dspes. %

\section{Conclusion}
Despite being a well-known problem in the literature, load balancing has not been exhaustively studied in the context of distributed stream processing engines.
Current solutions fail to provide satisfactory load balance when faced with skewed datasets. %
To solve this issue, we introduced \pkg, a new stream partitioning strategy that allows better load balance than key grouping while incurring less memory overhead than shuffle grouping.
Compared to key grouping, \pkgs is able to reduce the imbalance by up to several orders of magnitude, thus improving throughput and latency of an example application by up to 45\%. %

This work gives rise to further interesting research questions.
Is it possible to achieve good load balance without foregoing atomicity of processing of keys?
What are the necessary conditions, and how can it be achieved?
In particular, can a solution based on rebalancing be practical?
And in a larger perspective, which other primitives can a \dspe offer to express algorithms effectively while making them run efficiently?
While most \dspes have settled on just a small set, the design space still remains largely unexplored.

\bibliographystyle{IEEEtranN}

\bibliography{references}
\end{document}